\newif\ifjicv

\jicvfalse 

\ifjicv
\documentclass[pdflatex, sn-basic, Numbered, iicol]{sn-jnl}
\else 
\documentclass[conference]{IEEEtran}
\IEEEoverridecommandlockouts
\fi

\usepackage{graphicx}%
\usepackage{multirow}%
\usepackage{amsmath,amssymb,amsfonts}%
\usepackage{amsthm}%
\usepackage{mathrsfs}%
\ifjicv
\usepackage[title]{appendix}%
\else
\fi
\usepackage{xcolor}%
\usepackage{textcomp}%
\usepackage{manyfoot}%
\usepackage{booktabs}%
\usepackage{algorithm}%
\usepackage{algorithmicx}%
\usepackage{algpseudocode}%
\usepackage{listings}%

\newtheorem{definition}{Definition}
\newtheorem{proposition}{Proposition}
\newtheorem{corollary}{Corollary}
\newtheorem{remark}{Remark}

\usepackage{stmaryrd}

\newcommand{\descr}[1]{\vspace{0.2cm} \noindent \textbf{#1}}

\definecolor{citegreen}{HTML}{458B00}

\PassOptionsToPackage{hyphens}{url}\usepackage{hyperref}
\hypersetup{
   colorlinks=true,
   citecolor=citegreen
}

\begin{document}

\ifjicv
\title[Cryptography in Malware Obfuscation]{Use of Cryptography in Malware Obfuscation}

\author*[1]{\fnm{Hassan Jameel} \sur{Asghar}}\email{hassan.asghar@mq.edu.au}

\author[1]{\fnm{Benjamin Zi Hao} \sur{Zhao}}\email{ben\_zi.zhao@mq.edu.au}

\author[1]{\fnm{Muhammad} \sur{Ikram}}\email{muhammad.ikram@mq.edu.au}

\author[1]{\fnm{Giang} \sur{Nguyen}}\email{duclinhgiang.nguyen@hdr.mq.edu.au}

\author[1]{\fnm{Dali} \sur{Kaafar}}\email{dali.kaafar@mq.edu.au}

\author[2]{\fnm{Sean} \sur{Lamont}}\email{sean.lamont2@dst.defence.gov.au}

\author[2]{\fnm{Daniel} \sur{Coscia}}\email{daniel.coscia1@dst.defence.gov.au}

\affil[1]{\orgname{Macquarie University}, \orgaddress{\country{Australia}}}

\affil[2]{\orgname{Defence Science and Technology Group}, \orgaddress{\country{Australia}}}

\else

\title{Use of Cryptography in Malware Obfuscation$^*$\thanks{$^*$This is the full version of the paper with the same title to appear in the Journal of Computer Virology and Hacking Techniques.}}

\author{\IEEEauthorblockN{Hassan Jameel Asghar}
\IEEEauthorblockA{Macquarie University\\
Australia\\
Email: hassan.asghar@mq.edu.au}
\and
\IEEEauthorblockN{Benjamin Zi Hao Zhao}
\IEEEauthorblockA{Macquarie University\\
Australia\\
Email: ben\_zi.zhao@mq.edu.au}
\and
\IEEEauthorblockN{Muhammad Ikram}
\IEEEauthorblockA{Macquarie University\\
Australia\\
Email: muhammad.ikram@mq.edu.au}
\and
\IEEEauthorblockN{Giang Nguyen}
\IEEEauthorblockA{Macquarie University\\
Australia\\
Email: duclinhgiang.nguyen@hdr.mq.edu.au}
\and
\IEEEauthorblockN{Dali Kaafar}
\IEEEauthorblockA{Macquarie University\\
Australia\\
Email: dali.kaafar@mq.edu.au}
\and
\IEEEauthorblockN{Sean Lamont \\ and Daniel Coscia}
\IEEEauthorblockA{Defence Science and Technology Group\\
Australia\\
Email: sean.lamont2@dst.defence.gov.au\\
daniel.coscia1@dst.defence.gov.au}
}
\maketitle
\thispagestyle{plain}
\pagestyle{plain}
\fi

\ifjicv
\abstract{Malware authors often use cryptographic tools such as XOR encryption and block ciphers like AES to obfuscate part of the malware to evade detection. Use of cryptography may give the impression that these obfuscation techniques have some provable guarantees of success. In this paper, we take a closer look at the use of cryptographic tools to obfuscate malware. We first find that most techniques are easy to defeat (in principle), since the decryption algorithm and the key is shipped within the program. In order to clearly define an obfuscation technique's potential to evade detection we propose a principled definition of malware obfuscation, and then categorize instances of malware obfuscation that use cryptographic tools into those which evade detection and those which are detectable. We find that schemes that are hard to de-obfuscate necessarily rely on a construct based on environmental keying. We also show that cryptographic notions of obfuscation, e.g., indistinghuishability and virtual black box obfuscation, may not guarantee evasion detection under our model. However, they can be used in conjunction with environmental keying to produce hard to de-obfuscate version of programs.}  


\maketitle 
\else
\begin{abstract}
Malware authors often use cryptographic tools such as XOR encryption and block ciphers like AES to obfuscate part of the malware to evade detection. Use of cryptography may give the impression that these obfuscation techniques have some provable guarantees of success. In this paper, we take a closer look at the use of cryptographic tools to obfuscate malware. We first find that most techniques are easy to defeat (in principle), since the decryption algorithm and the key is shipped within the program. In order to clearly define an obfuscation technique's potential to evade detection we propose a principled definition of malware obfuscation, and then categorize instances of malware obfuscation that use cryptographic tools into those which evade detection and those which are detectable. We find that schemes that are hard to de-obfuscate necessarily rely on a construct based on environmental keying. We also show that cryptographic notions of obfuscation, e.g., indistinghuishability and virtual black box obfuscation, may not guarantee evasion detection under our model. However, they can be used in conjunction with environmental keying to produce hard to de-obfuscate version of programs.
\end{abstract}
\fi


\ifjicv
\keywords{Malware obfuscation, malware detection, cryptography, environmental keying}
\fi




\section{Introduction}
\label{sec:crypt-obfs}
Malware developers often obfuscate their programs in the hope that the program goes undetected by malware detectors, e.g., antivirus software. Many researchers have documented obfuscation techniques used in malware in the wild~\cite{aligot, hidden-str-obfs, MAIORCA201516}. 
Of particular interest is the use of cryptographic tools to obfuscate malware, e.g., encrypting parts of the program via a block cipher. Use of cryptography may give the impression that the said technique is provably resistant to de-obfuscation. However, as we shall see shortly, most use of cryptography by malware authors to obfuscate their programs amounts to {security via obscurity}. On the other hand some techniques can be shown to require substantial computing power to de-obfuscate~\cite{environ-keying}. A key question to ask is how do we evaluate if a malware obfuscation technique is resistant to de-obfuscation or not? In this paper, we propose a definition of malware obfuscation, and then categorise (cryptographic) malware obfuscation techniques, used by malware in the wild as well as proposed by malware analysts and computer security researchers, under this framework into those that are easily detectable versus those that avoid detection. {Note that detection here means the ability of the malware detector to detect the program as malware, and not to detect whether parts of the program are encrypted or not, which could be the case with both benign and malware programs.}




The reader may find some similarity of the topic to cryptographic program obfuscation~\cite{barak-impossibility}. Although cryptographic program obfuscation can be used in malware obfuscation, 
it is helpful to distinguish between the goals of malware obfuscation and that of program obfuscation in general. In case of the latter, given a program $P$, the goal is to create an obfuscated program $\hat{P}$ which is functionally equivalent to $P$ but which is harder to reverse-engineer~\cite{collberg1997taxonomy}. The programmer therefore wishes to hide how the program is implemented. On the other hand, the goal of malware obfuscation is to avoid being detected, and therefore labelled as malware, before it has run its functionality at least once on a target machine. Thus, simply obfuscating the program will not necessarily meet the goals of malware obfuscation, as the resulting program, by virtue of being functionally equivalent to the original program, can still be detected as malware through its input/output behavior, e.g., via dynamic analysis. Having said that, often times a malware author simply aims to avoid detection via inspection of the program code, i.e., via static analysis only. 

Perhaps the most common example of the use of cryptography in obfuscation is \emph{string obfuscation}~\cite{hidden-str-obfs}. Under this type of obfuscation, strings such as URLs, paths and constants are encrypted using a cryptographic cipher such as AES or DES~\cite{hidden-str-obfs}. Through string obfuscation, the malware author may wish to hide URLs which may have been flagged as malicious by anti-malware programs. Note that regardless of the security of the underlying cipher, e.g., DES vs AES, such obfuscation can in principle be undone. This is because the decryption logic of the program, including the decryption key, {is most often} provided within the application~\cite{hidden-str-obfs, wermke-obfs-google-play}. Thus, such strings can be de-obfuscated by a sufficiently sophisticated de-obfuscation tool or a skilled programmer using, for example, dynamic analysis. 

One of the goals of this paper is to decouple the use of cryptographic tools in the manner exemplified in string obfuscation, which essentially amounts to \emph{security by obscurity}, versus a more principled approach whereby robustness against de-obfuscation is guaranteed not merely because the detector was unable to locate the decryption/decoding routine within the program. At the outset we need to be clear about the opposing goals of malware obfuscation and detection. The goal of a malware author is to ensure that his/her program runs on the target machine. For this to happen, the malware should ``trick'' any anti-malware software running on the target system to believe it is a benign program. Obfuscation, then, is one of the techniques towards that end. On the other hand, malware defense would like to avoid such mistakes. On the precautionary side, we could label any use of obfuscation in a program as indicative of malicious intent. However, the fact that obfuscation is used by benign programs as well, e.g., to protect intellectual property~\cite{hidden-str-obfs} or {to defend against attacks on the software by creating \emph{metamorphic} copies~\cite{kazi2013metamorphic, rana2014metamorphic},} shows that this will inevitably block the execution of benign programs. This shows that we cannot simply label a program malware on the basis of obfuscation alone. A good detector should be able to distinguish between an obfuscation of a benign program versus that of malware. 

\descr{Our Contributions.} We give a formal model of malware obfuscation in Section~\ref{sec:formal}, where we define obfuscation to be successful if it degrades the combined false positive and negative rates of a malware detector. Under this formal model, we show that evasion is not possible if the obfuscated program is functionally equivalent to the original program. We then relax the requirement of functional equivalence to requiring the obfuscated program to run {its functionality} on selected target machines. We discuss a few of the prominent techniques from cryptography used in malware obfuscation, but which do not evade detection under our model in Section~\ref{sec:non-evasive}. We then give a description of an obfuscator which provably evades detection under our model in Section~\ref{sub:construct}, but which is not useful in the sense that the program's functionality may not even run on the target machine. We analyse a real-world alternative of such a scheme based on environmental keying~\cite{riordan-environ} from the lens of our model in Section~\ref{sub:envir}. In Section~\ref{sec:deny}, we discuss the use of deniable encryption in malware obfuscation and its shortcomings. Section~\ref{sub:formal-crypto} discusses cryptographic notions of obfuscation, such as indistinguishability obfuscation~\cite{barak-impossibility}, and its relation to malware obfuscation defined in this paper. We analyse the prevalence of cryptographic malware obfuscation in the real-world in Section~\ref{sec:real-world}, and discuss related work in Section~\ref{sub:rw}. Finally, we discuss some limitations of our treatise and avenues for future work in Section~\ref{sec:discuss}.    


\section{Formal Model and Implications}
\label{sec:formal}
A program $P$ is modelled as a probabilistic polynomial time Turing (PPT) machine. We consider a set $S$ of programs with two disjoint subsets: $\mathsf{Malware}$ and $\mathsf{Benign}$. We assume the two subsets to be mutually exclusive. We shall often use $M$ to denote a generic malware, i.e., a member of $\mathsf{Malware}$. We assume that $S$ is sampled from a joint probability distribution $\mathcal{D}$ of programs and their labels (malware or benign). We also assume that the set $S$ is of polynomial size. Thus, the set is not supposed to capture all possible benign and malicious programs, but rather a representative sample, on which we can test the capabilities of obfuscators and detectors. The assumption of $S$ being of polynomial size reflects the fact that the performance of any malware detection task is evaluated on a limited set of programs. Furthermore, we are also interested in knowing whether obfuscation can evade detection on programs that have already been labelled. A malware detector $D$, is a program, which takes as input a program $P$, and outputs 1 if it is malware and 0 otherwise. The type-I and type-II errors (false positives and false negatives, respectively), associated with $D$ are defined as:
\begin{align}
\label{eq:type-errors}
    \alpha_D &= \Pr[D(P) = 1 \mid P \in \mathsf{Benign}], \nonumber\\
    \beta_D &= \Pr[D(P) = 0 \mid P \in \mathsf{Malware}]. 
\end{align}
Here the probability is over the distribution $\mathcal{D}$ and any randomness employed by the detector $D$.\footnote{The notation $D(P)$ covers both static and dynamic analysis. In the former, the detector only makes its decision based on taking the description of the program $P$ as a string. In the latter, the program can run $P$ as a subroutine, feeding it with required inputs. This engulfs sandboxed environments.} Note that probability over the distribution $\mathcal{D}$ means that we sample a program uniformly at random from $S$ to calculate the two errors. When we talk about the total error or simply the error of a detector $D$, we mean the quantity $\alpha_D + \beta_D$. 
Without loss of generality, we can assume that $\alpha_D + \beta_D \leq 1$~\cite{how-far-can-we-go}. Because otherwise we have $\alpha_D + \beta_D > 1$, and we can instead use a detector which flips the output of $D$. The error of this detector is given by:
\[
1 - \alpha_D + 1 - \beta_D = 2 - (\alpha_D + \beta_D) < 1.
\]
We shall call $D$ trivial if it makes its decisions independent of the program, i.e., benign or malware. 
\begin{proposition}
\label{prop:trivial}
We have
\begin{enumerate}
    \item If $D$ is trivial then $\alpha_D + \beta_D = 1$.
    \item Conversely, if for any $D$, $\alpha_D + \beta_D = 1$, then there exists a trivial $D'$ such that $\alpha_{D'} = \alpha_D$ and $\beta_{D'} = \beta_D$.
\end{enumerate}
\end{proposition}
\begin{proof}
For part (1), assume $D$ is trivial. Let $P \in S$. Then:
\begin{align*}
    1 &= \Pr[D(P) = 1] + \Pr[D(P) = 0] \\
     &= \Pr[D(P) = 1 \mid P \in \mathsf{Benign}] \\
     &+ \Pr[D(P) = 0 \mid P \in \mathsf{Malware}] \\
     &= \alpha_D + \beta_D.
\end{align*}
For part (2), let $D'$ be a detector which with probability $\alpha_D$ outputs 1, else outputs 0. Then
\begin{align*}
     \alpha_{D'} &= \Pr[D'(P) = 1 \mid P \in \mathsf{Benign}] \\
     &= \Pr[D'(P) = 1] = \alpha_D,
\end{align*}
and 
\begin{align*}
    \beta_{D'} &= \Pr[D'(P) = 0 \mid P \in \mathsf{Malware}] = \Pr[D'(P) = 0] \\
    &= 1 - \alpha_D = \beta_D,
\end{align*}

as required.
\end{proof}
In light of the proposition, we shall call any $D$ with $\alpha_D + \beta_D =1$ as being trivial.
\begin{corollary}
\label{cor:non-trivial}
If $D$ is non-trivial then $\alpha_D + \beta_D < 1$.
\qed
\end{corollary}
In particular, any $D$ is trivial if $\alpha_D = 1$ or $\beta_D = 1$. For instance, if $S$ is a set of all programs that use an encryption algorithm, and $D$ is the malware detector that outputs 1 if a program uses an encryption algorithm. Then we have $\alpha_D = 1$, even though $\beta_D = 0$. This is exactly the example we highlighted in the preamble. The case when $\alpha_D = \beta_D = 0$ for all possible sets of programs $S$ is not possible since malware detection is undecidable~\cite{cohen-viruses, virus-impossibility}. Thus, in practice, any $D$ is expected to give a tradeoff. We shall assume that there is at least one non-trivial detector $D$ for the set $S$ of programs. Otherwise, obfuscation is pointless.

\begin{remark}
As mentioned, the programs $D$ and $P$ are PPT algorithms. This means that the size of $P$ is bounded by a polynomial. Also, if $D$ simply runs $P$ as a subroutine, it can only evaluate $P$ on polynomially many inputs from a possibly larger space.
\end{remark}

\descr{Functional Equivalence.} We say that two programs $P_1$ and $P_2$ are functionally equivalent if for all inputs $x$, we have $P_1(x) = P_2(x)$. Otherwise they are functionally inequivalent. We assume that any benign program is functionally inequivalent to a malware and vice versa. On the other hand, the benign programs (resp., malware) in $S$ may be functionally equivalent to one another. 

\descr{Malware Obfuscation.} A program obfuscator $\mathcal{O}$ is a PPT compiler that takes as input a program $P$ and outputs a program $P'$. We call $P'$ the obfuscation of $P$ under $\mathcal{O}$.
\begin{definition}
\label{def:mal-obf-all}
Let $\mathcal{O}$ be a malware obfuscator. Let $D$ be a malware detector. Define: 
\[
\beta^{\mathcal{O}}_D = \Pr[D(M') = 0 \mid M' \leftarrow \mathcal{O}(M), M \in \mathsf{Malware}]
\]
and,
\[
\alpha^{\mathcal{O}}_D = \Pr[D(P') = 1 \mid P' \leftarrow \mathcal{O}(P), P \in \mathsf{Benign}]    
\]
We say that $\mathcal{O}$ evades detection from the malware detector $D$, if 
$\alpha^{\mathcal{O}}_D \leq \alpha_D$ implies $\beta^{\mathcal{O}}_D > \beta_D$.
\qed 
\end{definition}

In other words, $D$ is successful if it simultaneously demonstrates a decrease in the false positive rate and an increase in the true positive rate on the obfuscated versions of the set of programs. This rules out cases where $D$ might show a decrease in the false positive rate but at the expense of its true positive rate, i.e., $\alpha^{\mathcal{O}}_D \leq \alpha_D$ and $\beta^{\mathcal{O}}_D \leq \beta_D$. {Furthermore, a malware detector might put more emphasis on reducing one of the two types of errors. For instance, an antivirus product may weigh reducing false positives more than false negatives. In this case, the definition states that keeping one type of error fixed, if the resulting obfuscation results in further degradation of the other error rate, then  the obfuscation technique is considered successful.} Definition~\ref{def:mal-obf-all} is formed after the definition of the ideal distinguisher in~\cite[\S 11.7]{cover-elements}. 

\descr{Trivial Obfuscators and Conservative Detectors.} The above definition is lenient: it calls the obfuscator a successful one if it evades detection, on average. Thus, for some subsets of malware it might perform worse. If it performs worse for all malware, then obviously the said technique is useless. The definition weeds out the trivial ``identity'' obfuscator, which simply prints out its input program as the output. This obfuscator will necessarily have $\alpha^{\mathcal{O}}_D = \alpha_D$ and $\beta^{\mathcal{O}}_D = \beta_D$, and hence it does not evade detection. 
More importantly, the definition penalises any brute-force or highly conservative way of detecting malware, by stating that there will necessarily be a tradeoff in the form of increased false positives. Consider for instance the set $S$ of all programs that do not use a cryptographic library. Let $\mathcal{O}$ be an obfuscator, that encrypts some components of the program. Let $D_2$ be the following malware detector that uses the (non-trivial) malware detector $D$ as a subroutine:

\begin{algorithmic}
\State $D_2(P)$:
\If{$P$ imports a cryptographic library}
    \State Output 1
\Else
    \State $D(P)$
\EndIf 
\end{algorithmic}

Clearly, we have $0 = \beta^{\mathcal{O}}_{D_2} \leq \beta_{D_2}$. However, we also have $1 = \alpha^{\mathcal{O}}_{D_2} > \alpha_{D_2}$, since $D$ is non-trivial and before obfuscation, none of the programs import a cryptographic library. Thus, this obfuscator evades detection against $D_2$ under our definition (as it should). Another detector in this line of detectors is the one that flags any use of obfuscation as an indicator that the program may be malicious. This may very well be the case in the real-world: obfuscation may be more prevalent in malware than benign programs. However, since we are interested in the strength of the obfuscation technique, and the capabilities of any detector against an obfuscation technique, we have modelled our definition as a challenge to the detector to distinguish the same technique applied to a malware versus benign program. Furthermore, benign programs also use obfuscation techniques for a variety of reasons such as protection of intellectual property~\cite{hidden-str-obfs, nate-mesh}. 

\descr{Utility and Functional Equivalence.} From a malware author's point of view, the above definition does not say anything about the utility of the obfuscated malware $M'$. For instance, the obsfuscator could simply remove all ``malicious functionality'' in $M$, and therefore (rightfully) achieve a higher error rate $\beta^{\mathcal{O}}_D$, saying nothing about the detection capabilities of $D$. On the one extreme, we may have functional equivalence. That is, $M(x) = M'(x)$ for all inputs $x$. However, this requirement means that there is no obfuscator which can evade detection from a malware detector which solely bases its decision on blackbox input-output behaviour (e.g., dynamic analysis), provided that functional inequivalence of programs in $S$ can be checked in polynomial time. That is, given a polynomial number of inputs, one can check whether two programs in $S$ are functionally equivalent or not. In case of a pair of functionally inequivalent programs, one obtains at least one input-output pair in polynomial time in which their output differs. We prove this in the following.

\begin{proposition}
\label{prop:func-pres}
Let $\mathcal{O}$ be an obfuscator that preserves functionality. If functional inequivalence in $S$ can be checked in polynomial time, then there exists a non-trivial detector such that $\mathcal{O}$ does not evade detection from it.
\end{proposition}
\begin{proof}
Let $D$ be a non-trivial detector for $S$ guaranteed by assumption. This means that $\alpha_D + \beta_D < 1$. The idea is to construct a detector $D'$ that runs $D$ over the (unobfuscated) programs in $S$, marks its decisions, and then creates input-output fingerprints for each program in $S$. After obfuscation is applied, since the programs remain functionally equivalent, $D'$ will be able to identify them through the input-output fingerprints, and hence retain the error rate of $D$ (as the labels are from $D$). The only subtlety is if some programs in $S$ are functionally equivalent, which we address in the following. 

We construct a detector $D'$ from the set $S$ of programs and labels, and the detector $D$ as follows. For each program in $S$, $D'$ runs $D$, and stores the label output by $D$. Next, at each round, $D'$ chooses as input $\{0, 1\}^q$, where $q$ is a polynomial. $D'$ then runs each program in $S$ on this input, storing the respective outputs.\footnote{If the programs are probabilistic, we assume that the input contains the input to the random tape, i.e., coin tosses~\cite[\S 7.1]{arora2009computational}.} After at most polynomially many inputs $Q$, the sequence of input-output pairs of each program labelled as benign is different from the sequence of input-output pairs of each program labelled as malware, according to the assumption on functional inequivalence. However, as noted before there may be some benign programs that are functionally equivalent to other benign programs. The same goes for malware programs. 

As long as $D$ labels them uniformly, there is no issue. However, if $D$ labels a subset of programs with the same input-output sequence differently, then $D'$ will not be able to distinguish between their obfuscated versions given only input-output sequences. Let $S_1, S_2, \ldots, S_t$ be subsets of benign programs from $S$, such that each program in $S_i$ is functionally equivalent, and for each $i$, not all programs in $S_i$ have the same label under $D$. Take one such $S_i$. Since each program in $S_i$ is benign, let $j(i)$ be the number of programs in $S_i$ mislabelled by $D$ as malware. Similarly we have a collection of subsets of malware programs such that within each subset they are functionally equivalent, but $D$ mislabels at least one of them as benign. Due to symmetry, the ensuing analysis is applicable to this case as well. Therefore, without loss of generality, we use the benign case only. Let $S(t) = S_1 \cup S_2 \cup \cdots \cup S_t$. Note that it does not matter if each set $S_i$ is truly a set of benign or malicious programs, as $D'$ is simply replicating the behaviour of $D$ on them. In other words, $D'$ does not need to know if $S_i$ is a set of benign or malicious programs.

After the programs in $S$ are run through $\mathcal{O}$, our detector $D'$ does as follows. Given any program $P$, it runs the program on the inputs from $Q$. If the input-output sequence identifies it as a program not in $S(t)$, $D'$ simply outputs the label previously stored from $D$. Otherwise, it identifies the set $S_i$ such that $P \in S_i$. Note that the programs in different $S_i$'s have different input-output ``fingerprints.'' After this, $D'$ returns the label 1 (malware) with probability $j(i)/|S_i|$, else it outputs the label 0 (benign). We now calculate the type-I errors of the two detectors. Since $D$ is run once by $D'$ (before obfuscation), let us calculate its \emph{empirical} type-I error $\alpha_D^{\text{Emp}}$ defined as:
\begin{align*}
    \alpha_D^{\text{Emp}} &= \frac{|D(P) = 1 \mid P \in \mathsf{Benign}|}{|\mathsf{Benign}|} \\
    & = \frac{|D(P) = 1 \mid P \in \mathsf{Benign}, P \notin S(t)|}{|\mathsf{Benign}|} \\
    & + \frac{|D(P) = 1 \mid P \in \mathsf{Benign}, P \in S(t)| }{|\mathsf{Benign}|}\\
    & = \frac{j + \sum_{i = 1}^t j(i)}{|\mathsf{Benign}|},
\end{align*}
where $j$ is the number of benign programs not in $S(t)$, mislabelled by $D$. Now, for $D'$ we have:
\begin{align*}
    \alpha^{\mathcal{O}}_{D'} &= \Pr[D'(\mathcal{O}(P)) = 1 \mid P \in \mathsf{Benign}] \\
        &= \Pr[D'(\mathcal{O}(P)) = 1 \mid P \in \mathsf{Benign}, P \notin S(t)] \\
        &\times \Pr [P \in \mathsf{Benign}, P \notin S(t)] \\
        &+ \Pr[D'(\mathcal{O}(P)) = 1 \mid P \in \mathsf{Benign}, P \in S(t)] \\
        &\times \Pr [P \in \mathsf{Benign}, P \in S(t)] \\
        &= \frac{|D(P) = 1 \mid P \in \mathsf{Benign}, P \notin S(t)|}{|\mathsf{Benign}| - |S(t)|} \\
        &\times \frac{|\mathsf{Benign}| - |S(t)|}{|\mathsf{Benign}|} 
        + \sum_{i = 1}^t \frac{|S_i|}{|\mathsf{Benign}|} \cdot \frac{j(i)}{|S_i|} \\
        &= \frac{j}{|\mathsf{Benign}|} + \frac{\sum_{i = 1}^t j(i)}{|\mathsf{Benign}|} = \alpha_D^{\text{Emp}}
\end{align*}

Thus, $D'$ reproduces $D$'s type-I error. A similar analysis holds for type-II error, i.e., $\beta^{\mathcal{O}}_{D'}$. Hence, we conclude that $D'$ is non-trivial and hence $\mathcal{O}$ cannot evade detection from it. 
\end{proof}

We reiterate that the theorem only holds if programs in $S$ can be checked to be functionally inequivalent in polynomial time. This is needed to check input-output behaviour on the same set $S$ of programs before and after obfuscation. Functional equivalence in general is an undecidable problem. Thus, the statement does not hold for general programs. However, such a guarantee is outside the control of the obfuscator $\mathcal{O}$, and only the property of the set $S$. We are instead interested in the capabilities of the obfuscator in making malware programs evade detection without relying on the specific set $S$.
In light of the above proposition, we have a milder (and realistic) requirement that there is some target subset $X$ of inputs $x$, for which we have $M'(x) = M(x)$. Given this we have the following definition of utility.

\begin{definition}[Utility]
\label{def:util}
For a program $P$, let $X(P)$ denote a target set of inputs; a subset of the domain of $P$. An obfuscator is useful if for all programs $P$ and target input sets $X(P)$, we have $P'(x) = P(x)$ for all $x \in X(P)$. If $X(P)$ is exactly the input domain of $P$, we say that the obfuscator preserves functionality.
\end{definition}

The above definition means that for some inputs (not in the target set), the obfuscated malware $M'$ might not even behave maliciously. However, there should be at least one input on which the program exhibits malicious behaviour for the program to be labelled as a malware. This naturally models the use of environmental variables to identify target machines on which the malware is supposed to run (as opposed to every machine)~\cite{riordan-environ}. We will discuss this technique in Section~\ref{sub:envir}. 

\descr{On the Obfuscator.} The obfuscator itself is benign, in the sense that if we run the obfuscator on any benign program, the resulting obfuscation is not malicious (in the global sense, irrespective of the point of view of any detector). The utility aspect is important, because in principle we could write a malware which evades detection against any (polynomial-time) detector as we shall show in Section~\ref{sub:construct}. This essentially means that the malware may never run its functionality on its target. 

\section{Non-Evasive Obfuscation Techniques}
\label{sec:non-evasive}
We first summarise some techniques that do not evade detection in our model. The main reason being that the decryption routine is shipped with the program. The obfuscation technique is simple: encrypt part(s) of the malware using an encryption algorithm. The encrypted components are decrypted at run-time. Thus, such obfuscation (in principle) is detectable both via static and dynamic analysis, since detection amounts to finding the decryption routine within the program, running it to decrypt the encrypted components and analyzing them in the clear. In our model, this would mean that the error of the detector will remain the same after such obfuscation is applied. Since these are the predominant techniques used by real-world malware authors~\cite{MAIORCA201516}, we nonetheless enlist some of them. Example usage of these techniques are in string encryption~\cite{hidden-str-obfs, wermke-obfs-google-play}. Here, only the string (such as a URL) is encrypted. Once again, the decryption key is provided within the program so that the string can be decrypted at run-time. Similarly, they can be used in class encryption: e.g., DEX file encryption~\cite{wermke-obfs-google-play}: this again suffers from the same aforementioned problem. These techniques in essence are similar to the use of ``packing'' in malware, which uses compression to evade detection. Needless to say that decompression is done at run-time.

\subsection{Base-64 Encoding}
\label{subsub:base64}
Base64 encoding is an encoding scheme that converts binary data into text. A non-malicious use of Base64 encoding is sending images over email. Base64 encoding is also frequently used to obfuscate malware, e.g., to obfuscate file names and file content~\cite{apvrille-crypto}. This is related to cryptography only at a rudimentary level, as Base64 encoded text, once detected, can readily be decoded. This is obvious, since this is merely an encoding scheme and not an encryption scheme. The following is an example of Base64 encoding in JavaScript taken from~\cite{apvrille-crypto}:
\begin{verbatim}
this.jdField_b_String 
= a(b("L1RodW1icy5kYg=="));
\end{verbatim}
Here, \verb+b()+ Base64 decodes the input string, and \verb+a()+ reads the resource. The decoded string is \verb+Thumbs.db+ which itself contains Base64 encoded malicious data~\cite{apvrille-crypto}.

\subsection{XOR Obfuscation}
The underlying idea behind XOR obfuscation is its {potential to be used as a one-time pad}. Given a plaintext represented in bits, if it is XORed with a random key of equal length, then the resulting encryption is perfectly secure. If the plaintext represents a malware, then one can encrypt the malware by using a key of equivalent length. However, this requires huge key sizes. As a result, malware obfuscators normally reuse a short key, e.g., a single byte, and encrypt equivalent sized blocks. The use of a single byte is known as single-byte XOR encryption~\cite[\S 13]{prac-mal-analysis}. Due to short keys used in XOR encryption based obfuscation, there are various tools that can deobfuscate the program or find the key. These tools look for expected text at a given location in the program, in an attempt to find the key (e.g., PE files)~\cite[\S 13]{prac-mal-analysis}. Finding the key does not mean that these detectors necessarily try to look for the decryption routine within the malware to find the key~\cite[\S 13]{prac-mal-analysis}. 
For a list of few other notable schemes that belong to the this category of encryption using short keys, we refer the reader to~\cite[\S 13]{prac-mal-analysis}.

\subsection{Obfuscation with Stronger Encryption}
Many obfuscation tools also offer obfuscation via block ciphers such as AES, DES and TEA~\cite{hidden-str-obfs, aligot}. However, even though block ciphers such as AES in an appropriate mode can provide strong encryption, their use in malware obfuscation is rudimentary in our model. The decryption routine together with the key is normally part of the code which obfuscates the code~\cite{malware-authors-dont-learn, hidden-str-obfs}. Note that these ciphers are also employed to encrypt traffic between the malware and a remote server (command and control traffic)~\cite{aligot}. But once again since the decryption keys are hardcoded, the traffic can be decrypted to see what information is being exchanged~\cite{malware-authors-dont-learn}. One could also use asymmetric encryption algorithms, e.g., RSA, to generate session keys to encrypt communication as in the case of the Waledac malware family~\cite{malware-authors-dont-learn}. The RSA public-private key is generated at run-time by the malware, and hence the private key is also readily available to the detector (via dynamic analysis). For a detailed account of block ciphers used by commercial obfuscation tools to obfuscate Android and Java apps in the manner described above, please see~\cite{hidden-str-obfs}.

\subsection{Use of Hash Functions}
At first glance, it is not clear how cryptographic hash functions could be used in obfuscation, as their one-wayness property implies that deobfucation would be impossible even for the malware itself. The use of hash functions such as MD5 and SHA in malware have been documented for generating unique identities of command and control bots such as in the Waledac malware family~\cite{hidden-str-obfs, aligot}. This is obviously not an application of obfuscation. However, as we show next, one can use hash functions to construct a scheme that provably evades detection in our model, albeit with an unavoidable utility tradeoff.  

\section{Towards Evasive Techniques: The Hash-then-Decrypt Construct}
\label{sub:construct}
We now turn our attention to obfuscation techniques that evade detection (under our model). The construct given in this section is derived from a similar construct based on environmental variables which first appears in~\cite{riordan-environ}. {The construct given in this section is theoretical in nature as it is not useful according to our definition. However, we present it here because its security is the basis of a practical construct based on environment variables to be presented in the next section.} Let $\text{Enc}$ and $\text{Dec}$ be the encryption and decryption functions of a semantically secure symmetric key encryption scheme. Let $H$ be a cryptographic hash function. Given the program $P$ and its block $B$, to be obfuscated, the obfuscator $\mathcal{O}$ does the following:
\[
 k \leftarrow \text{random key}, \llbracket B \rrbracket \leftarrow \text{Enc}_k(B), \, h_k \leftarrow H(k),
\]
where the key is selected from the key space of the symmetric encryption scheme. The obfuscator then creates the program $P'$, which is the same as $P$ except that $B$ is replaced by the block:

\begin{algorithmic}
\State $x \leftarrow \text{KeyFinder()}$
\If{$H(x) = h_k$}
    \State $B' \leftarrow \text{Dec}_x(\llbracket B \rrbracket)$ 
    \State $B'$
\EndIf 
\end{algorithmic}
Here the routine KeyFinder() simply samples a random key from the key space. We now discuss the security and utility of this obfuscator. 

\begin{proposition}
\label{prop:ran-key-sec}
If the symmetric key cryptosystem is semantically secure, then the random key-based hash-then-decrypt obfuscator $\mathcal{O}$ evades detection against any detector $D$ in the random oracle model.
\end{proposition}
\begin{proof}
Let $S$ be a set containing two programs: one benign and other malware, denoted $P$ and $M$, respectively. We assume that $P$ and $M$ are identical except for the blocks $B_P$ and $B_M$ in $P$ and $M$, respectively. This implies that $M$ is benign if the block $B_M$ is replaced with $B_P$. Let $D$ be a non-trivial detector achieving errors $\alpha_D$ and $\beta_D$, respectively, with both strictly less than 1. Let $\mathcal{O}$ denote the above obfuscator. Let $\mathcal{A}$ denote the semantic security adversary. Adversary $\mathcal{A}$ chooses $m_0 = B_P$ and $m_1 = B_M$ as its choices of the two messages.\footnote{We assume they are of the same length. Otherwise, we pad the shorter program with more benign code, e.g., print statements.} When $\mathcal{A}$ receives $\llbracket m_b \rrbracket$, it inserts this as $\llbracket B \rrbracket$ in the program. Adversary $\mathcal{A}$ then gives the resulting program $P'$ to $D$. Whenever $D$ makes a hash query, the adversary samples a uniform random string and gives it to $D$. For repeat hash queries, the adversary returns the same random string as previously sent. Let $b'$ be the bit returned by $D$. Adversary $\mathcal{A}$ outputs $b'$ as its guess. It is clear that the adversary perfectly simulates the obfuscation except for negligible probability (e.g., in case $D$ finds a collision). Now, assume that the obfuscator does not evade detection against $D$. This means that from Definition~\ref{def:mal-obf-all}, we have simultaneously: 
\[
\alpha^{\mathcal{O}}_D \le \alpha_D \text{ and } \beta^{\mathcal{O}}_D \le \beta_D.
\]
Then, except for negligible probability, we have:
\begin{align*}
    \Pr[ b' = b ] &= \Pr[b' = 1 \mid b = 1] \Pr[b = 1] \\
    &+ \Pr[b' = 0 \mid b = 0] \Pr[b = 0] \\
    &= (1 - \alpha^{\mathcal{O}}_D) \cdot \frac{1}{2} + (1 - \beta^{\mathcal{O}}_D) \cdot \frac{1}{2} \\
    &\geq (1 - \alpha_D) \cdot \frac{1}{2} + (1 - \beta_D) \cdot \frac{1}{2}\\
    &= 1 - \frac{1}{2}(\alpha_D + \beta_D) > \frac{1}{2},
\end{align*}
where the last step follows due to the non-triviality of $D$, i.e., $\alpha_D + \beta_D < 1$ (Corollary~\ref{cor:non-trivial}). This is clearly a contradiction, as $\mathcal{A}$ breaks the semantic security of the cryptosystem.
\end{proof}

Unfortunately (or fortunately), this obfuscator is not useful. 
\begin{proposition}
\label{prop:ran-key-util}
The random key-based hash-then-decrypt obfuscator $\mathcal{O}$ is not useful.
\end{proposition}
\begin{proof}
Pick any input $x \in X(M)$ in the target set of inputs. On this input, necessarily (unobfuscated) $M$ runs the malicious block $B$ (otherwise it is not the target set of inputs). On the other hand, on the same input, the probability that the obfuscated version $M'$ will run the block $B$ is negligible (in the size of the key space). Through the union bound, over polynomially many inputs in $X(M)$, the probability that $M'$ will run the block $B$ remains negligible. Hence with overwhelming probability $M'(x) \neq M(x)$ for all inputs in $X(M)$, and hence the malware obfuscator is not useful.
\end{proof}


\section{Environmental Keying}
\label{sub:envir}
While the obfuscated malware construction in Section~\ref{sub:construct} is not useful (according to our definition), variants of it have appeared in theory and practice. On the one extreme, we have the cryptographic key uniformly at random selected from a key space, which means that the decryption routine cannot be run in feasible time, by the malware as well as the detector. On the other extreme, the malware writer could provide the key within the obfuscated program. However, under our model, this technique is not secure, as it is a matter of time before some detector will be able to detect keys.\footnote{This is one of the reasons why XOR-based obfuscation techniques are easily detected and de-obfuscated by common program analysis tools.} In between these two extremes we have a spectrum of difficulties: 
\begin{enumerate}
    \item Use of a random cryptographic key so that the KeyFinder routine just takes long enough to evade dynamic analysis~\cite{oxpat}. This means use of smaller key-sizes than is considered cryptographically secure. This also includes use of other time-consuming options, e.g., hashing multiple times~\cite{oxpat}. However, for the obfuscated malware to be useful, the KeyFinder should output the key in feasible time. Hence, this technique does not evade detection as the correct key will be found on all machines (the obfuscated malware is functionally equivalent).
    \item Instead of embedding it in the hash, the key can be retrieved from an external source, e.g., a server~\cite{oxpat, hash-cond, prac-mal-analysis}. {Note that in this case a hash function is not used at all.} However, the link to the server itself may disclose maliciousness, e.g., a known malicious domain. On the other hand, there are examples of malware who retrieve keys from trusted hosts, such as GitHub, Dropbox, and Google Docs, thus evading such reputation-based detection~\cite{cisco-2018, anomali-rise, sean-sophos}). Furthermore, a growing trend in malware is communicating over TLS/SSL with remote command-and-control hosts~\cite{anomali-rise}. Encrypted communication could also be used to download part of the malware hosted on mainstream cloud platforms, such as GitHub, which may not raise suspicion~\cite{sean-sophos}. However, if a trusted host is used to host the key and the key is communicated to the malware via an encrypted channel, the detector will still be able to label the program as malware using dynamic analysis, as it will fetch the key on any machine. {We also note that it is still possible to detect malicious intent even if the traffic is encrypted by using machine learning algorithms trained on features extracted from encrypted benign and malware traffic such as the size of the network flow~\cite{shekhawat2019encryptedmaltraffic}}. Thus, such techniques only delay the inevitable. 
    \item By far the most resilient technique is environmental keying~\cite{riordan-environ}. Instead of brute-forcing the key, the program constructs the key from the environment variables (e.g., user name). The main point being that the malware would run (after a mini-bruteforce search) on target computers (whose environment variable values are known beforehand). However, on non-target machines, the KeyFinder routine may not terminate as the environment variables never take on the pre-determined values ``hidden'' within the hash digest. 
\end{enumerate}

The technique of environmental keying has been known at least since 1998~\cite{riordan-environ}. We present a detailed account of this technique in the hash-then-decrypt construct in the related work section (Section~\ref{sub:rw}). Here we look at the use of environmental keying with the help of an example.

\subsection{The Ebowla Framework: Case Study}
\label{subsub:ebowla}
We consider the Ebowla framework~\cite{ebowla} for encrypting malware payloads via environmental keying which has been presented at several white hat offensive security conferences. As mentioned earlier, environmental keying uses values stored in ``environment variables'' to derive the key. The idea being that the attacker knows the values taken by these environment variables in his/her target machines, and hence reproducing the key will be faster on the target machines versus other machines (including those employed by malware detectors as sandbox environments). The framework describes the following environment variables, although it leaves room for writers to define more variables:
\begin{itemize}
    \item Environmental variables: e.g., username, user domain, computer name, and number of processors.
    \item Path variables: e.g., \verb+C:\Windows\temp+. This includes a starting location, e.g., \verb+C:\Windows+.
    \item External IP ranges, e.g., \verb+100.0.0.0+.
    \item System time range with the grarnualrity of year, month or a particular day, e.g., 20210000, 20211200, 20211201.
\end{itemize}
One or more of these variables, called tokens, are then used to create the key $k$ in the hash-then-decrypt construct. An important consideration here is how much logic is pre-built in the hash checking part to ensure a reasonable tradeoff between key reconstruction by the malware, versus by the detector. For instance, if the key is derived by constructing one variable each from the above mentioned variables, and the hash checking routine constructs values for these tokens in sequence, then this reduces the \emph{entropy} of the key space, versus if the hash checking routine needs to check all possible combinations (order of concatenation). We shall return to the entropy of the key space in environmental keying shortly. For now, we focus on the particular version of hash-then-decrypt construct, employed in the Ebowla framework. The obfuscation routine is as follows:
\begin{align*}
    k &\leftarrow H(\text{environmental key}), \llbracket B \rrbracket \leftarrow \text{Enc}_k(B), \\
    & h_B \leftarrow H(B),
\end{align*}
where as before we assume that $B$ is the malicious block that needs to be obfuscated (the so-called payload). The obfuscated program $P'$ then contains the following block instead of $B$:
\begin{algorithmic}
\State $x \leftarrow \text{KeyFinder}(\text{environment variables})
$
\State $B' \leftarrow \text{Dec}_x(\llbracket B \rrbracket)$
\If{$H(B') = h_B$}
    \State $B'$
\EndIf 
\end{algorithmic}
where 
\begin{algorithmic}
\State $\text{KeyFinder}(\text{environment variables})$:
\State environmental key $\leftarrow$ Concatenate values of environment variables 
\State \Return $H(\text{environmental key})$
\end{algorithmic}

Note that the above technique can be used both by a malware (to obfuscate its malicious payload) or a benign program, e.g., to check if only the rightful user is able to run the protected part of the program~\cite{nate-mesh}. Two main differences between this hash-then-decrypt construct and the one discussed in Section~\ref{sub:construct} are that (a) the key is obtained as the hash of the constructed environmental variable string -- this is obviously done to increase entropy of the key, and (b) the hash of the block $B$, i.e., $h_B$, is also provided in the program $P'$. This second difference, however, means that in our model the obfuscator does not evade detection. The proof is simple: any non-trivial detector $D$ on the unobfuscated set of programs $S$ can keep hashes of the benign and malicious blocks. Recall that the set $S$ is of polynomial size. Thus, the detector retains its advantage (type-I and type-II errors) over the obfuscated variants of the two types of programs. In practice, what this means is that any detector that keeps signatures of past programs will be able to detect this obfuscation. To avoid this, the authors suggest that not all of the block be used for hashing (by using an offset, e.g., by discarding the last few bytes before hashing). This can then evade pre-computed signatures. However, the offset is part of the decryption routine, and hence can be used by a detector to recompute hashes for malware that have already been detected. This ``vulnerability'' can be removed if instead we use the following routine based on the hash-then-decrypt construction of Section~\ref{sub:construct}. Namely, the obfuscator first constructs:
\begin{align*}
    k &\leftarrow H(\text{environmental key}), \llbracket B \rrbracket \leftarrow \text{Enc}_k(B), \\
    & h_k \leftarrow H(k),
\end{align*}
and then replaces the block $B$ with:
\begin{algorithmic}
\State $x \leftarrow \text{KeyFinder(\text{environment variables})}$
\If{$H(x) = h_k$}
    \State $B' \leftarrow \text{Dec}_x(\llbracket B \rrbracket)$ 
    \State $B'$
\EndIf 
\end{algorithmic}
Notice the double application of the hash function, so that $h_k = H(k) = H(H(\text{environmental key}))$. The program does not store $k$. Only $h_k$ and the encrypted block $\llbracket B \rrbracket$ are hardcoded into the program. 

\descr{Remark.} The construction from Ebowla shown above is a simplification of the one shown in~\cite{ebowla}, as we ignore compression and encoding of the encrypted block, which is likely there to increase portability of the payload, rather than for security reasons. One important aspect however is how the environment key is derived and how it is reconstructed. 

\subsection{Finding the Environmental Key -- Target vs Non-Target Machines}
The main idea behind environmental keying is that finding the key should be easy on the target machines, but hard if the detector does not know the target machines. There is an implicit assumption that the environmental variables used are sufficiently unique (high entropy) among different machines. This problem makes more sense if the attacker's target is a much smaller set of machines than the total pool of machines. Without knowing the identity of the target machines (and hence the state of their environment variables), a detector's only choice may be to run a brute-force search on the environment variables used. 

Translated to our formal model, the detector $D$, given the obfuscation $P'$ of a program using the environmental-key based hash-then-decrypt obfuscation technique of Section~\ref{subsub:ebowla}, would like to flag it as a malware or benign program. The detector has no advantage in classifying it as one or the other without finding the environmental key. To model environmental variables and their values taken up by machines, we consider a universal set $\mathcal{C}$ of all environmental values taken up by any machine. The set is bestowed with a probability distribution, modeling the probability that a given environmental variable can take on a specific value. We assume this distribution to be public, and hence also known to the detector. At each time step, we assume that the detector can sample one profile at a time and hence potentially run the given obfuscated program on this environmental profile.\footnote{This can of course be generalized to polynomially many profiles.} Let $\mathcal{T}$ be the target set of profiles from $\mathcal{C}$. That is, those profiles, on which the environmental key matches the hash, and therefore decryption is successful. Let $p_{\mathcal{T}}$ denote the probability of sampling such profiles. We assume that these profiles are hardwired into $\mathcal{O}$ (via the hash digest), and hence the same target profiles work for all programs obfuscated by $\mathcal{O}$.  

We are interested in finding:
\[
\beta^{\mathcal{O}}_D = \Pr[D(M') = 0 \mid M' \leftarrow \mathcal{O}(M), M \in \mathsf{Malware}]    
\]
and,
\[
\alpha^{\mathcal{O}}_D = \Pr[D(P') = 1 \mid P' \leftarrow \mathcal{O}(P), P \in \mathsf{Benign}],    
\]
where we assume that $D$ is given access to $\mathcal{C}$. Furthermore by assumption, $D$ is non-trivial, i.e., $\alpha_D + \beta_D < 1$.  
The observation is that if $p_\mathcal{T}$ is a non-negligible function in the size of the profiles, then the detector will find the key in polynomially many samples.\footnote{We use the usual definition of a negligible function, i.e., one that grows slower than the reciprocal of any polynomial.} In this case, both error rates are $\beta^{\mathcal{O}}_D = \beta_D$ and $\alpha^{\mathcal{O}}_D = \alpha_D$, and hence the obfuscator does not evade detection. On the other, hand, if $p_\mathcal{T}$ is a negligible function in the size of the profiles, then the detector cannot find the key in polynomially many samples. The detector then outputs $1$ with some fixed probability $p$, else it outputs $0$. This implies that $\alpha^{\mathcal{O}}_D + \beta^{\mathcal{O}}_D = (1 - p) + p = 1$, regardless of $p$. Combining with $\alpha_D + \beta_D < 1$, we see that if $\alpha^{\mathcal{O}}_D \le \alpha_D$, then necessarily $\beta^{\mathcal{O}}_D > \beta_D$, and hence $\mathcal{O}$ evades detection against $D$. 

Thus, from the malware author's point of view, the goal is to ensure that $p_\mathcal{T}$ is negligible. One way to achieve this is to ensure that the set of profiles $\mathcal{C}$ has high entropy. That is, there is a large number of possible values that can be taken up by the accessed environmental variables across all machines.

\descr{On the Key Finder Routine.} The Key Finder routine takes as input a set of environment variables and ``extracts'' a string which is the purported key to be used in the hash-then-decrypt construct. The use of the word extract is intentional, as this may not be a simple concatenation of the environment variables. For instance, in the Ebowla framework, the Key Finder routine starts by loading the current value in the \texttt{PATH} variable, and then traverses the file system from this value (outputting the current value at each point). Note that we are assuming that the key finding routine is itself not obfuscated. We will return to the case when this routine may be obfuscated as well in Section~\ref{sub:formal-crypto}. For now, we assume it to be in the clear. This means that the malware obfuscator cannot hide the environmental variable(s) used to extract the key, as well as the entire routine.




\descr{Environment Variables and Entropy.} We can define environment variables as showing the current state of the system. They can be categorized into two main types:
\begin{enumerate}
    \item \emph{Time invariant:} These remain static over time. Examples include: \texttt{USER}, \texttt{HOME}, \texttt{PATH}, and IP address. 
    \item \emph{Time variant:} These change over time. A simple example is current system time.
\end{enumerate}
For entropy, an important consideration is how many computers take on a particular value of an environment variable. Consider a universe of $n$ computers. Let $A$ be an environment variable, and let $N(A = a)$ denote the number of computers having the state $A = a$ (having their environmental variable $A$ set at $a$). Denote by $P_A(a) = N(A = a)/n$ as the empirical probability of having the environment variable $A$ set to $a$. Then, we can see that if $A$ is the system time, and $a$ is the current time, then $P_A(a) = 1$. On the other hand if $A$ is the username (the variable \texttt{USER}) then $P_A(a)$ is close to 0. 
A key difficulty is to have estimates of these empirical probabilities in the real-world, as there is a lack of datasets due to obvious privacy concerns. However, we can guess the entropy through what is known about similar variables via other means. For instance:
\begin{itemize}
    \item \textit{User:} In many organizations, computer login names are a combination of the first and the last name to ensure uniqueness across the organisation. In the US alone, there were 6,299,033 unique surnames according to the 2010 US Census~\cite{surnames-2010}. A different study curated a list of 4,250 unique first names~\cite{firstnames}. The combination of these two alone gives a total space of $\approx 2^{35}$. This is arguably a crude lower bound as it is only confined to the US. Furthermore, computer names are more involved, as they may include abbreviations and/or additional characters. 
    However, if the malware is targeting a group of users with similar user names, then the search space reduces drastically. As we have mentioned earilier, the malware will not be able to hide the fact that it is extracting a certain subset of user names.
    \item \textit{IP Addresses:} There are a total of $2^{32}$ IP addresses under IPv4. The use of IPv6 will further increase this space. Malware authors might be interested in a certain IP prefix, in which case the space is reduced.
    \item \textit{System Time:} Time-based activation may start based on the granularity of a second, a minute, an hour, a day, a week, or a year. This gives a granularity of about $2^{25}$ a year. Once again, broadening the range of time when the malware is triggered reduces the space.
\end{itemize}
Thus, a combination of all three increases the search space for the detector. But this increase in entropy comes at a tradeoff: the higher the entropy, the more targeted the attack, and the less likely it will be detected by the detector. On the other hand, the lower the entropy, the wider the net cast by the malware, but the easier it is for the detector to identify it as malware. 


\descr{Time on Non-Target Machines.} Given that the Key Finder routine is in the clear, is it possible to make the routine take longer on non-target machines without making it take longer on the target machine regardless of the number of targets? For instance, in the Ebowla framework described above, the code traverses the file system to extract strings, as well as tries all combinations of the different environment variables. Obviously, this is done on both the target and other machines, and hence the excess time also penalizes the target machine. The answer to this question is mixed. Consider the following construct:

\begin{enumerate}
    \item Denote by ${a}$ the string obtained by concatenating the values of one or more environmental variables in the target computer (denoted by $A$).
    \item Let $m$ be a large positive integer. Pick a \emph{small} number $i$ from $\{1, 2, \ldots, m\}$.
    \item The environmental key is $H(a || i)$.
    \item The Key Finder routine is then: concatenate the values of the environment variables $A$, and concatenate it with $j = 1, 2, \ldots, m$ until a hash match is found.
\end{enumerate}

The idea here is that on any machine that does not exhibit $a$ as the concatenated value of the corresponding environment variable, the Key Finder routine will likely run through all values until $m$, by which time the detector will discard the machine as not being the target. If $m$ is large enough, this can be substantial time. On the other hand, $i$ can be chosen to be small so that Key Finder returns the key after only going till $i$ on the target machine. A drawback of this construct is that the detector might only run the Key Finder for a feasible number of steps before discarding the current machine as not being the target machine (as the intention of the malware writer would not be to run indefinitely on the target machine). Even if we make $m$ exponentially large, if the number of possible environmental values are small (low entropy, or in other words, the size of profiles $\mathcal{C}$ is polynomially bounded), then the detector can run the Key Finder routine in parallel for several profiles at a time. In our Turing machine language, this means using dovetailing~\cite{lovasz-dovetail} to output the key in polynomial time.

Another possibility is to accept a certain trade-off. Finding the key on the target machine may take time $t$, but this will also increase the work of the detector by at least $t$-folds. In the Ebowla framework, this is done by making the Key Finder go through combinations of environmental variables, or traversing through the file system without specifying depth. However, a simpler construct based on hash puzzles (from cryptocurrencies) may do the same trick~\cite{bitcoin-book}. In this case, we can build the puzzle so that it takes ``on-average'' a certain amount of time before the key is found on the target machine. On the non-target machine, this will take roughly the same amount of time per computer profile by the detector before it can decide the search to be futile, as the detector knows the possible number of combinations. A benefit of this approach is that the program author does need more environmental variables to increase search space of the detector.

This argument also shows that it is not possible to devise a Key Finder that will make the detector take exponential time while running for polynomial time on the target machine. This follows from the dovetailing trick~\cite{lovasz-dovetail}. As long as the detector has the target profile in its polynomially sized set of profiles, the detector will eventually find the key by running the same routine. Thus, the only possibility is to ensure that the number of profiles is of exponential size. This can be achieved only via increasing the entropy of the environmental variables used. 



\section{Deniable Encryption}
\label{sec:deny}
Environmental keying does raise suspicion in the sense that a detector knows the presence of encrypted code, even if the encrypted payload may not be decrypted by the detector.\footnote{Indeed, presence of high entropy code can be flagged as indication of malicious intent~\cite{filiol-entropy}.} But what if there are multiple possible decryptions of the same encrypted payload, and only one of them malicious? If the malware detector is only able to decrypt the benign versions then it is likely to flag the resulting program as benign. On the other hand, on the target computer, if it is more likely that the malicious version will be decrypted, then we have a malware that potentially avoids detection, albeit in a more steganographic sense. An encryption system that has such a property is called \emph{deniable encryption}~\cite{deniable-encryption}.

The possible use of deniable encryption in malware has been discussed in~\cite{filiol-malicious, ebowla}. A simple XOR-based scheme illustrates the concept. Suppose we have a benign program $P$ and a malware $M$. Abusing notation, we use $P$ and $M$ to denote the corresponding strings. We also assume that $|P| = |M|$; otherwise, one can use padding to make the two equal. Then, one constructs a key $K_1$ of length equal to $|P| = |M|$ and then computes:
\[
C = M \oplus K_1,
\]
and further computes:
\[
K_2 = C \oplus P.
\]
By the property of XOR, we have 
\[
C = M \oplus K_1 = P \oplus K_2.
\]
Thus, $C$ could either be the encryption of $M$ under the key $K_1$ or the encryption of $P$ under the key $K_2$. Note that, since $K_1$ is random, so too is $K_2$ (without conditioning on the other). The idea is to deny $C$ being an encryption of a malware, and presenting the key $K_2$ if need be to show that the underlying payload was $P$~\cite{filiol-malicious}.

There are some issues with this approach:
\begin{itemize}
    \item The key length needs to be the same as the malware program, which can be quite big. Alternative deniable encryption techniques are presented in~\cite{deniable-encryption}, based on both public-key and shared-key cryptography. However, they are rather inefficient in terms of length of ciphertext per plaintext. One such example is given below.
    \item How does the malware author guarantee that the correct key is used in the target computer? One way is to use an environmental key. But since the key $K_2$ is dependent on the first key, this would make finding the second key hard on non-target computers as well, which defeats the purpose of deniable encryption. 
    \item Another way of providing the two keys is to host them on a server, which is then retrieved by the program at run-time (e.g., HTTP keying). The malware author initially only uploads the key $K_2$, and later at a specific time it posts $K_1$. Thus, the detector will not be able to decrypt the malware block until the relevant key ($K_1$) is uploaded.
    \item There is dependency between the two keys. Once the key $K_2$ is obtained, $K_1$ is not random. An ideal technique (from the malware author's point of view) should be where the two keys are independent. 
    \item The scenario of deniable encryption is not entirely suited to the obfuscation. Deniable encryption assumes a sender and a receiver. The adversary can coerce either one of them to reveal the key. There are also sender-and-receiver deniability schemes, but they assume some intermediaries through which the encrypted communication is transmitted. In the obfuscation scenario, the sender and receiver are both the obfuscated program. Therefore, some public-key encryption schemes described in~\cite{deniable-encryption} do not work in this scenario as the trapdoor information has to be in the code.
\end{itemize}

Another symmetric key encryption scheme from~\cite{deniable-encryption} can be summarised as follows: the obfuscator encrypts a benign program $P$ and malware $M$ separately with two encryption keys. The obfuscated program contains a routine to search for keys (say they are environmental keys). The hash-then-decrypt construct matches both hashes, and if any of them match, it decrypts the program matching the hash. Given that in most computers, the benign routine is always decrypted, the malware writer may use plausible deniability. However, the unencrypted code will always raise suspicion. And the detector would not hastily label this as benign. 

\section{Cryptographic Notions of Obfuscation}
\label{sub:formal-crypto}
There are some key differences between the goal of cryptographic obfuscation and real-world use of cryptographic tools in malware obfuscation. The former considers a program $P$, and requires its obfuscated version, $\hat{P}$, to be functionally equivalent, while at the same time (informally) nothing should be learned about $P$ from the description of $\hat{P}$. However, this does not directly translate to the goal of malware obfuscation (as defined in this paper), which aims to evade detection. In particular, since $\hat{P}$ preserves functionality, the detector $D$ can still label it as malware based on its input-output behaviour, as we shall see shortly. The above informal description of cryptographic obfuscation captures the concept of virtual black box (VBB) obfuscation, defined as follows.

\begin{definition}[VBB Obfuscation~\cite{barak-impossibility}]
\label{def:vbb}
A program obfuscator $\mathcal{O}$ is a virtual black box (VBB) obfuscator, if it preserves functionality and for all PPT algorithms $D$, there exists a PPT simulator $R$ such that for all programs $P$
\begin{equation}
\label{eq:vbb}
\left| \Pr \left[ D(P', 1^\lambda) = 1\right] - \Pr \left[ R^{P}(|P|, 1^\lambda) = 1\right] \right| < \text{negl}(\lambda) 
\end{equation}
where $P' \leftarrow \mathcal{O}(P, 1^\lambda)$.
\end{definition}

In the above, $\lambda$ is the security parameter, and we require any algorithm taking it as input to run in polynomial time in its size as well. The simulator $R$ replaces the algorithm $D$ in the black-box setting, and the superscript $P$ indicates that $R$ only has black-box access to the program $P$. By now it is well-known that VBB obfuscation is impossible in general~\cite{barak-impossibility}. As a result, alternative ``weaker'' notions have been proposed of which the notion of indistinguishability obfuscation (iO) shows the most promise.

\begin{definition}[iO~\cite{barak-impossibility}]
\label{def:io}
A program obfuscator $\mathcal{O}$ is an indistinguishability obfuscator if it preserves functionality, and for all functionally equivalent programs $P_1$ and $P_2$ of equal size
\begin{equation}
\label{eq:io}
\left| \Pr \left[ D(P'_1, 1^\lambda) = 1\right] - \Pr \left[ D(P'_2, 1^\lambda) = 1\right] \right| < \text{negl}(\lambda) 
\end{equation}
where $P'_i \leftarrow \mathcal{O}(P_i, 1^\lambda)$, for $i = 1, 2$.
\end{definition}

Note that functionally equivalent programs means that $P(x) = P'(x)$ for all inputs in the domain of $P$ (and $P'$). Our informal statement above can now be stated formally.

\begin{proposition}
\label{prop:vbb-io}
Let $\mathcal{O}$ be a VBB or iO obfuscator. Then there exists a set $S$ of malware and benign programs, and a detector $D$ such that $\mathcal{O}$ does not evade detection from $D$.
\end{proposition}
\begin{proof}
Let $D$ be the detector of Proposition~\ref{prop:func-pres}, and let $S$ be a set of programs whose functional inequivalence can be checked in polynomial time as in Proposition~\ref{prop:func-pres}. Then, since $\mathcal{O}$ preserves functionality, the result follows from Proposition~\ref{prop:func-pres}. 
\end{proof}





As we discussed earlier, the obfuscator cannot do anything if the set $S$ has the property that all its programs can be checked for functional inequivalence in polynomial time. Therefore, the only way we can get around the limitation posed by the above proposition is if we let go of the ``functionality preserving'' requirement of the obfuscated program. Namely, we only want the program to act as a malware on a limited (target) set of computers (see Definition~\ref{def:util}). In this case, one may devise a malware that runs as a malware only if the current computer's name equals a predefined list of names (a use of environmental keying). However, there does not seem to be much benefit of obfuscation here over the hash-then-decrypt construct. If the program reads the environmental variables, it is necessary interacting with its environment, and hence we can learn the variables accessed by it through dynamic analysis. This brings us back to the same issue with environmental keying and the hash-then-decrypt construct.

\descr{Obfuscating the Key Finder Routine.} But the advantage here is that the Key Finder routine itself can be obfuscated. For instance, the Key Finder routine could read a bunch of environment variables on the machine, and only use a small subset of them to decide whether to trigger the malware behavior or not. The remaining variables can be discarded. In this case, since the Key Finder routine is obfuscated, the detector will not know which environmental variable is needed, and hence load a larger environmental profile than is actually being used by the program. 

\descr{On VBB and iO.} Cryptographic obfuscation, both VBB and iO, is a widely studied problem. There are works that show that VBB obfuscation is still possible under \emph{idealized models}, e.g., the generic graded encoding model~\cite{vbb-generic-graded}. On the iO end, various constructions have beep proposed via multi-linear maps (or graded encoding schemes) including~\cite{ggh-13, clt-13, ggh-15}. Unfortunately, many of these schemes have been broken, e.g.,~\cite{ggh-15}. An account appears in~\cite{albrecht2020multilinear}. There are also constructions of iO schemes not based on multilinear maps or graded encoding schemes, e.g.,~\cite{jain2021wellfounded}. Finally, there are also implementations of iO schemes~\cite{apon2014implementing} as reported in~\cite{xu2020layered}, and~\cite{implementing-ggh-15}. Thus, while this is still a field with ongoing research, a successful and practical scheme may have an impact in terms of malware obfuscation, e.g., by obfuscating the Key Finder routine in the hash-then-decrypt construct.

\section{Obfuscation using Cryptographic Tools in the Wild}
\label{sec:real-world}


In Table~\ref{tab:obf-sum} in Appendix~\ref{app:sum} we summarize properties of the various malware obfuscation techniques using cryptographic tools discussed in this paper. This includes methods to detect them, as well as real-world tools to apply these techniques. We are also interested in knowing how frequently cryptographic obfuscation is applied in the wild. {This is important to relate the relevance of our categorisation of malware obfuscation using cryptographic tools in the real-world. Namely, we are interested in knowing whether a large percentage of real-world malware are using environmental keying or other encryption algorithms to encrypt part of the malware.} Datasets showing the prevalence of cryptographic obfuscation in malware are few and far between, partly because it is difficult to programmatically infer if a cryptographic library is being used for obfuscation or for other tasks such as encrypting communication, or encrypting local files (in the case of ransomware). In Section~\ref{sub:rw}, we discuss some works in literature that have surveyed obfuscation techniques in the wild. However, since their coverage is not limited to obfuscation via cryptographic techniques (evasive or otherwise), we, therefore, run our own analysis to get an idea of cryptographic malware obfuscation. 
As a first step, we need to have a ground-truth dataset to identify malware that use cryptographic obfuscation. 
For this, we leverage labelled samples (binaries) from the Dike dataset.\footnote{Dike Dataset: \url{https://github.com/iosifache/DikeDataset}} The labels for this dataset are obtained by parsing the categories returned by antivirus vendors aggregated by VirusTotal. Among the categories is that of an ``encryptor,'' which corresponds to programs that use obfuscation. Taking this smaller dataset as ground truth, we extract EMBER features~\cite{ember} to train a machine learning model to predict cryptographic obfuscation over the much larger SOREL dataset~\cite{sorel}. {EMBER features is a list of 2381 features extracted from the program binaries and includes features extracted from the header, imported functions as well as format agnostic features such as byte histogram. These features are automatically extracted from program binaries using EMBER's open-source code repository~\cite{ember}.}    

The specific label from the Dike dataset we use to characterize programs using cryptographic obfuscation are 6,197 samples flagged as ``encryptors'' and not ``ransomware.'' The remaining 3,735 samples are considered as not using cryptographic obfuscation. Our decision to exclude ransomware from encryptors is due to the high likelihood that ransomware will use encryption libraries for tasks other than obfuscation. 
With a training/testing split of 80\%/20\%, a random forest classifier with 100 trees is trained with the EMBER features as inputs and the labels defined above. The resulting classifier yields an AUC of 0.966 in distinguishing between the two defined classes: encryptors and non-encryptors. The confusion matrix with the breakdown of original classes and errors is provided in Figure~\ref{fig:dike_model_cm}. With these results we are satisfied in the ability of the classifier to infer the encryptor label on the unlabelled EMBER features of the SOREL dataset. We sample 200,000 Malware labelled samples, and 200,000 Benign labelled samples from the SOREL dataset. 
The classifier identified 10.9\% (21,834) of the Malware samples as using encryption obfuscation. On the other hand, 6.5\% (13,087) of Benign samples were also flagged as using encryption obfuscation.

\begin{figure}[t]
\centering
\includegraphics[width=0.65\linewidth]{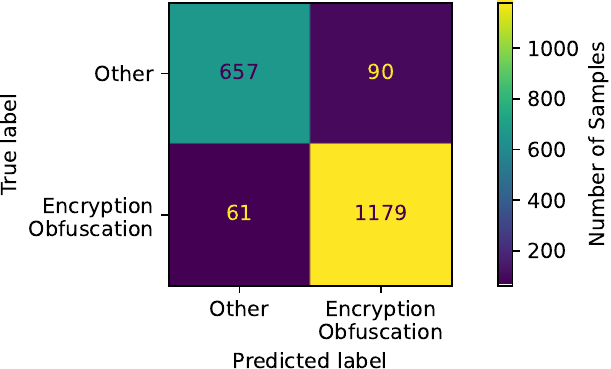}
\caption{Confusion matrix for cryptographic obfuscation detection on the test set from the Dike Dataset; AUC was 0.966, Accuracy was 92.4\%.
}
\label{fig:dike_model_cm}
\end{figure}

\begin{figure}[t]
\centering
\includegraphics[width=0.80\linewidth]{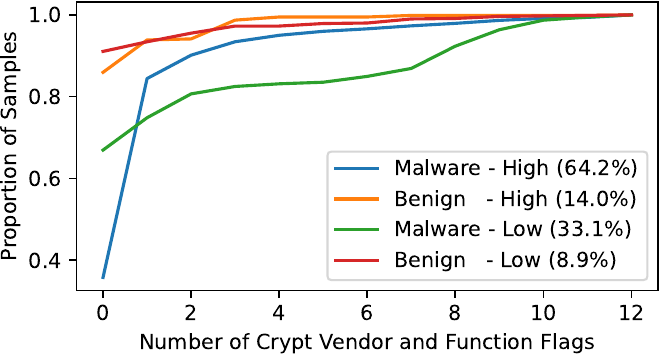}
\caption{Cumulative distribution of SOREL samples with VT reports demonstrating encryption keywords found in imported function names and AV vendor reports. The four groups consist of 1000 samples of the Highest and Lowest probability of using cryptographic obfuscation.
The percentage shown is proportion of samples containing a search term.
}
\label{fig:crypt_func_vendor_sorel}
\end{figure}

To validate these predictions on the SOREL dataset, we analyze a smaller sample of 1,000 samples each with the highest and lowest confidence values labelled by the classifier as using cryptographic obfuscation, from both Malware and Benign samples as labelled in the SOREL dataset. With these 4,000 samples, we obtain a VirusTotal (VT) report detailing the sample's behaviour and consider heuristics of their behaviour that may have warranted the classification. VT is an anti-virus aggregator; when a user submits a binary for analysis, VT scans the binary with a selection of anti-virus vendors. If more than a threshold number of these  vendors flag the binary as malicious, VT will report the sample as malicious.
Within each sample's VT report, each vendor also reports a short indicator of what sort of maliciousness exists. For example, 
Kaspersky may report  \texttt{Packed.Win32.Krap.iu}, which indicates that the sample belongs to the family of malware which protects against reverse engineering, i.e., obfuscation. Additionally, the VT analysis extracts any imported functions from the binary. From both these sources, we use a list of search terms that may indicate obfuscation using encryption. This is the same list used in Dike Dataset and includes: ``crypt'', ``cryp'', ``coder'', ``pack'', and ``krypt''. 

The cumulative distribution of samples which contained VT reports are depicted in Figure~\ref{fig:crypt_func_vendor_sorel}. More than 64\% of the Malware samples labelled as encryptors with high confidence contain at least one vendor labelling it as using cryptographic obfuscation, compared to 33\% among the malware samples labelled as such with low confidence. This sits in contrast with benign samples; only 8.9 and 14\% were flagged as using cryptographic obfuscation by any vendor. This discrepancy is likely due to the ``safe'' nature of benign samples: AV vendors 
not detecting any concerning behaviour, will not label them as belonging to any malware family; hence no ``crypt'' search term. In contrast, our classifier operating on the EMBER features may have captured encryption behaviour otherwise not flagged by the AV vendors. Nevertheless we are confident that the classifier correctly detects encryptors in malware, indicating that 10\% of them as using some form of cryptographic obfuscation.

\section{Related Work}
\label{sub:rw}

\descr{Enviromental Keying.} The concept of environmental keying was first proposed by Riordan and Schneier in~\cite{riordan-environ}. Further investigation of the topic in the academic circle appears in~\cite{hash-cond} and the work on secure triggers~\cite{secure-triggers}. A related concept is used in the mesh design, hash-then-decrypt method proposed by Nate Lawson~\cite{nate-mesh}. This uses the same concept of secure triggers albeit for different applications including gaming (unlocking higher levels, if the user has collected certain items), as well as protecting software. Glynos looks into different ways in which environmental keys can be derived~\cite{glynos-context}. Offensive security practitioners have also looked into environmental keying~\cite{ebowla, leo-derbycon}. Most importantly, the technique has been used in real-world malware to encrypt part or whole of the payload.\footnote{See for instance: \url{https://attack.mitre.org/techniques/T1480/001/}} The most notable example being the Gauss malware~\cite{gauss}. To the best of our knowledge, the encrypted payload of this malware, detected in 2011, has still not been decrypted~\cite{jumping-air-gap}. The malware consisted of a reconaissance component which would collect system information. The system information was then used by attackers to derive environmental keys~\cite{jumping-air-gap}. As noted in~\cite{jumping-air-gap}, the payload could only be decrypted by the unique target in an air-gapped network, and therefore, researchers have failed to decrypt it even after trying millions of combinations and publicly releasing technical details. Such sophisticated use of malware obfuscation can most likely be attributed to state-based actors. In general, encryption may also be used for other purposes, such as concealing the communication with a remote server~\cite{sean-sophos}. A report by Cisco states that an analysis of 400,000 malware binaries revealed about 70\% of them as using some encryption in 2017~\cite{cisco-2018}.


\descr{Theoretical Modeling.} Some researchers have also looked at formalising the concept of programs using environmental keying and the properties they must satisfy. Futoransky et al~\cite{secure-triggers} proposed the notion of secure triggers, which is similar to the environmental keying based hash-then-decrypt construct first proposed in~\cite{riordan-environ}. The difference being that instead of using the hash function to check if the key is valid, they use the encryption of the all 0 string as a predicate. They show that the construct is secure under the universal composability framework~\cite{canetti2001universally}, in the sense that the adversary (detector in our case) does not learn the contents of the encrypted block if the predicate is from a large enough space and under the semantic security of the encryption scheme. Blackthorne et al.~\cite{environ-keying} define security of environmental keying under an analyst (detector) with different capabilities, such as not knowing the target versus after the malware has infected its target. They also look at malware that could change its behaviour after sensing if its environment has been altered by the analyst (e.g., sandboxed environments)~\cite{environ-sense}. The authors of~\cite{malware-rerand} look at malware that derives its environment keys not just from the target but also from the nodes in the network in its path. Decryption is only successful if the correct path order is followed. In this paper, we have proposed a definition independent of the cryptographic technique being used for obfuscation. Our model defines the detector (variously, the analyst or the adversary) in terms of its ability to distinguish between malware and benign programs, instead of pinning it to environmental keying.




\descr{Surveys on Malware Obfuscation.} Several works have studied the prevalence of obfuscation techniques used by malware in the wild. You and Yim~\cite{you2010malware} give a brief taxonomy of malware obfuscation techniques, amongst which they mention encrypted malware. Such malware consists of an encrypted part and a decryptor. The drawback of such malware from the authors' point of view is that antivirus scanners may detect such malware through the commonly used decrypting routine. As mentioned in our paper, we also consider this technique as non-evasive, but mainly due to the fact that malware can be detected at run-time. Aligot~\cite{aligot} is a tool that identifies malware samples that obfuscate the use of block ciphers, hash functions and public key encryption. The authors test their tool on a couple of known families of malware that use cryptographic tools. However, their main experiments are done on synthetically generated samples. The tool may not be able to differentiate the use of cryptographic functions for obfuscation versus for other purposes. The authors in~\cite{hidden-str-obfs} propose the StringHound tool to detect string obfuscation. They enlist some encryption algorithms used for string encryption, e.g., XOR and AES encryption. As the name suggests, their survey is limited to string obfuscation. Several obfuscation techniques for Android malware are discussed in~\cite{MAIORCA201516}, of which string and class encryption are related to the topic of this SoK. The goal of their survey however is to show the effect of these obfuscation techniques on the unobfuscated programs and whether or not anti-malware engines are able to detect them. Wermke et al look at different obfuscation techniques employed on Android apps via popular obfuscation software~\cite{wermke-obfs-google-play}. Related to cryptography are string and class obfuscation, which we categorise as non-evasive techniques. A layered taxonomy of obfuscation techniques is presented by Xu et al \cite{xu2020layered} to help developers protect their software from intellectual theft in a more systematic way. Among the techniques 
they mention the hash-then-decrypt construct from~\cite{hash-cond} as well as cryptographic obfuscation (indistinguishability obfuscation). However, they do not view the techniques in light of their efficacy in evading detection against a model such as the one proposed in our paper. 

\descr{Other Work Related.} There are some other works worth mentioning in the use of cryptographic tools in malware. Some of these works are modifications of environmental keying, and are therefore left out as they retain the basics of that construct. For instance, the Bradley virus~\cite{filiol-malicious} uses nested encryption, where the key for the first encryption is derived from the environment variable. If the hash of this derived key matches the hardcoded hash, then the program continues to decrypt the next segment of the program. Otherwise it deletes the whole malware. {A related use of multiple encryptions is presented in~\cite{di2016transcriptase}. The main issue handled here is that a single encryption layer to encrypt the payload increases the entropy of the program which could therefore be flagged as malware via entropy analysis. The authors' solution is to use a substitution cipher to reduce the entropy of the first encryption layer, followed by a third encryption layer which converts the output of the second layer into a code-like format.} Perhaps, the most notable area of work marrying cryptography and malware is cryptovirology~\cite{young1996cryptovirology}, which deals with the study of malware that use cryptography to encrypt operations on the victim's computer. Ransomware are the most notable example of such malware. However, note that the use of cryptography in such programs is not meant for obfuscation, as encryption is used mainly for denying access to files with only the malware author having the means to decrypt them. Similar to the spirit of the steganographic nature of the deniable encryption scheme, the authors in~\cite{suarez2015stegomalware} discuss the use of steganography to hide malicious executable components in the assets directory of smartphone apps, e.g., multimedia files. As the authors note this only promises evading detection via static analysis as these malicious components would be retrieved during run-time. Another area of work related to the topic of this paper is white-box cryptography~\cite{chow2003white}. The predominant goal of white-box cryptography is security against a powerful adversary who is in control of the execution environment of a cryptographic program (such as a block cipher) and obtains an implementation of it with an embedded key. Unfortunately, most if not all implementations that purportedly provide protection against key extraction from a white-box attacker have been shown to be broken~\cite{bock2020white}. Note that even a successful white-box implementation will not protect a malware from evasion via dynamic analysis. 





\section{Limitations and Open Problems}
\label{sec:discuss}

\begin{itemize}
\item Our definition of malware detection does not take distributional aspects into account. For instance, cryptographic libraries may be used more by malware than benign programs. Thus, the definition does not rule out detectors who base their decisions on such distributions. Likewise, mere presence of any obfuscation may indicate that the program is more likely to be malware rather than benign. However, our definition challenges the detector in being able to differentiate the same obfuscation technique being applied on malware and benign programs, which is arguably more principled. 

\item The only obfuscation techniques using cryptographic tools that are shown to be hard in this paper are hash-then-decrypt construct and deniable encryption using environmental keys. Both use environmental keys and therefore can only evade detection if targeting only a small subset of machines. This appears to be a fundamental limitation of malware obfuscation. It appears that any technique that runs the malicious payload on a large number of machines after a small amount of time can in principle be detected; the detector only needs to wait for sufficiently long before it is able to detect malware.  

\item On the experimental side, while we have shown some results on the real-world use of cryptographic obfuscation by malware programs, this needs to be further validated and refined into different classes of techniques for obfuscation. In particular, we do not know how prevalent is the use of environmental keying in malware obfuscation.

\item Many real-world malware obfuscation techniques are only used to avoid specific types of malware analysis, e.g., use of packing to avoid static analysis, or assessing the run-time environment to detect whether a program is being dynamically analysed in a virtual environment. In our paper, we have not considered these specifics as these techniques do not provide provable guarantees of evading detection. 

\end{itemize}



\section*{Ethical Considerations}
All instances of malware obfuscation detailed in this paper are in the public domain. As such, we have not introduced any malware obfuscation technique that may jeopardize the security of systems. The aim of this paper is to simply categorise them from the point of view of difficulty of detection.

\ifjicv
\section*{Disclaimer}
\else
\section*{Acknowledgements}
\fi
This work was partially supported by the Australian Defence Science and Technology (DST) Group under the Next Generation Technology Fund (NGTF) scheme. The datasets generated and analysed in this paper are available from the corresponding author on reasonable request.

\ifjicv
\else
\bibliographystyle{IEEEtran}
\fi
\bibliography{crypto-ref}

\ifjicv
\begin{appendices}
\else
\appendices
\fi



\section{Summary of Cryptographic Malware Obfuscation Techniques}
\label{app:sum}
\ifjicv
\begin{table*}[]
\caption{Cryptographic techniques for obfuscating malware, tools that apply them, and methods to detect them.}
\label{tab:obf-sum}
\centering
\resizebox{\textwidth}{!}{
\begin{tabular}{c | c | c | c | c} 
\toprule%
\multirow{2}{*}{Method} & \multirow{2}{*}{Purpose} & \multicolumn{2}{|c|}{Identification*/Detection} & \multirow{2}{*}{Obfuscation Tools}  \\
\cmidrule{3-4}
& & Static & Dynamic & \\
 \midrule%
 String Encryption & Hide C\&C URLs~\cite{apvrille-crypto} & Entropy analysis*~\cite{lyda2007entropy, hidden-str-obfs} & Slicing~\cite{hidden-str-obfs} & DexGuard~\cite{wermke-obfs-google-play} \\ 
 & Hide class/method names~\cite{wermke-obfs-google-play} & Cryptographic libraries*~\cite{hidden-str-obfs} & Sandbox/debugging~\cite{afianian2019malware} & DexProtector~\cite{wermke-obfs-google-play}\\
 \hline
 Class Encryption & Hide implementation & Entropy analysis* & Sandbox/debugging &  DexGuard \\
 & & Cryptographic libraries* & & DexProtector \\
 \hline
 Environmental Keying & Hide payload & Entropy analysis* & Brute-force & Ebowla~\cite{ebowla}  \\
 & & Cryptographic libraries* &  & \\
 & & Brute-force~\cite{afianian2019malware} & & \\
 \hline
 Packing & Hide/compress payload & Entropy analysis*  & Sandbox/debugging & UPX~\cite{lyda2007entropy} \\
  & & Signature \& unpacking &  & \\
 \hline
Stegomalware & Hide malware executables & Statistical tests~\cite{provos2001detectingstego} & Sandbox/debugging & F5~\cite{f5, suarez2015stegomalware} \\
\botrule
 \end{tabular}
}
\end{table*}
\else
\begin{table*}[!h]
\caption{Cryptographic techniques for obfuscating malware, tools that apply them, and methods to detect them.}
\label{tab:obf-sum}
\centering
\resizebox{\textwidth}{!}{
\begin{tabular}{c | c | c | c | c} 
\hline%
\multirow{2}{*}{Method} & \multirow{2}{*}{Purpose} & \multicolumn{2}{|c|}{Identification*/Detection} & \multirow{2}{*}{Obfuscation Tools}  \\
\cline{3-4}
& & Static & Dynamic & \\
 \hline\hline%
 String Encryption & Hide C\&C URLs~\cite{apvrille-crypto} & Entropy analysis*~\cite{lyda2007entropy, hidden-str-obfs} & Slicing~\cite{hidden-str-obfs} & DexGuard~\cite{wermke-obfs-google-play} \\ 
 & Hide class/method names~\cite{wermke-obfs-google-play} & Cryptographic libraries*~\cite{hidden-str-obfs} & Sandbox/debugging~\cite{afianian2019malware} & DexProtector~\cite{wermke-obfs-google-play}\\
 \hline
 Class Encryption & Hide implementation & Entropy analysis* & Sandbox/debugging &  DexGuard \\
 & & Cryptographic libraries* & & DexProtector \\
 \hline
 Environmental Keying & Hide payload & Entropy analysis* & Brute-force & Ebowla~\cite{ebowla}  \\
 & & Cryptographic libraries* &  & \\
 & & Brute-force~\cite{afianian2019malware} & & \\
 \hline
 Packing & Hide/compress payload & Entropy analysis*  & Sandbox/debugging & UPX~\cite{lyda2007entropy} \\
  & & Signature \& unpacking &  & \\
 \hline
Stegomalware & Hide malware executables & Statistical tests~\cite{provos2001detectingstego} & Sandbox/debugging & F5~\cite{f5, suarez2015stegomalware} \\
\hline\hline
 \end{tabular}
}
\end{table*}
\fi

We summarize the main types of cryptographic malware obfuscation techniques discussed in this paper in Table~\ref{tab:obf-sum}. We also highlight some of the reasons why they are applied, methods to detect them, and obfuscation tools that apply these techniques. We differentiate between techniques that can only identify an obfuscated program (identification) versus those that can also detect whether the obfuscated program is malware or not (detection). The techniques in the identification category are marked with a `*' in the table. For instance, entropy analysis~\cite{lyda2007entropy, hidden-str-obfs} or the detection of certain cryptographic functions and libraries~\cite{hidden-str-obfs} are good static analysis methods to identify if some strings in a program have been obfuscated. However, they themselves are not sufficient to decide if the program is malware or not, since they do not decrypt the string itself. 

For the case of environmental keying (based on the hash-then-decrypt construct mentioned in this paper), a brute-force strategy can be used both statically or dynamically to find the environmental key. This is because the detector (malware analyst) only needs to find the pre-image of the hash, which can be done offline. We note that the hash-then-decrypt construct falls under the general category of \emph{trigger}-based malware~\cite{afianian2019malware}. In general, the trigger value may not be encrypted or hashed. In such a case, there are automated techniques that identify triggers (e.g., looking for an if-then-else construct) and find trigger values that initiate these triggers~\cite{brumley2008automatically, crandall2006temporal, moser2007exploring, afianian2019malware}. However, they do not work if the triggers are encrypted or hashed~\cite{moser2007exploring}. Conspicuous by its absence in the table is the deniable encryption scheme (see Section~\ref{sec:deny}). This is because the efficacy of the scheme depends on how it is applied. For instance, if the keys for both benign and malware decryption are shipped with the program then it can be readily detected as malware or benign using dynamic analysis. On the other hand, if it is used in the hash-then-decrypt construct, then it gives no benefit over environmental keying. 

Packing is a technique to compress a binary to reduce its size. But it has also been adopted by malware authors to make it hard to analyze their programs by static analysis. Although compression is not encryption, this technique is often lumped in with encoded/encrypted malware~\cite{kang2007renovo} since it renders the program incomprehensible without unpacking. Hence we have decided to include it in the table. One of the ways to detect a packed malware, is to look for signatures of the packing algorithm, e.g., UPX~\cite{lyda2007entropy}. Once this is identified, unpacking can be done to reveal the nature of the underlying program. Note that this can be done using static analysis. Finally, the techniques mentioned for detecting stegomalware (see Section~\ref{sub:rw}), i.e., statistical tests, and to produce stegomalware, i.e., the F5 steganographic algorithm~\cite{f5}, are in fact general-purpose for image steganography, and are not specifically made for malware steganography. 

\ifjicv
\end{appendices}
\else
\fi

\end{document}